%% file: ArXiv_GP1.tex
\documentclass{amsart}
\usepackage{amsmath,amssymb}
\usepackage{pstricks, pst-plot, pst-node, pst-grad, pst-math}
\usepackage{graphicx,xcolor}
\usepackage{arydshln}

 \newtheorem{thm}{Theorem}[section]
 
 \newtheorem{lem}[thm]{Lemma}
 \newtheorem{prop}[thm]{Proposition}
 \theoremstyle{definition}
 
 \theoremstyle{remark}

 \numberwithin{equation}{section}
 \usepackage{xcolor}
 \usepackage{mathabx}

\newcommand\CH{{\mathcal H}}
\newcommand{\R}{\mathbb{R}}

\newcommand{\C}{\mathbb{C}}

\newcommand\CD{{\mathcal D}}
\newcommand{\domS}{\mathcal{D}}

\newcommand{\dom}{\mathop{\rm dom}}
\newcommand{\Ran}{\mathop{\rm ran}}

\newcommand{\sqt}{{\mathfrak{t}}}
\newcommand{\Div}{\rm div}

\begin{document}

\title[Spectra of Gurtin-Pipkin type of integro-differential equations]{Spectra of Gurtin-Pipkin type of integro-differential equations and applications to waves in graded viscoelastic structures}

%----------Author 2
\author{Christian Engstr\"om} %corresponding author
\address{Department of Mathematics, Linnaeus University, Sweden}
\email{christian.engstrom@lnu.se}



\date{}
\dedicatory{}

\begin{abstract}
In this paper, we study spectral properties and spectral enclosures for the Gurtin-Pipkin type of integro-differential equations in several dimensions. The analysis is based on an operator function and we consider the relation between the studied operator function and other formulations of the spectral problem. The theory is applied to wave equations with Boltzmann damping.
\end{abstract}

%%% ----------------------------------------------------------------------
\maketitle
%%% ----------------------------------------------------------------------
\input{S1_Introduction.tex}
\input{S2_spectrum.tex}
\input{S3_equivalence.tex}

\input{S4_Applications.tex}

%\input{S5_additional.tex}

\vspace{3.5mm}

%{\small
%{\bf Acknowledgements.} \ 
%The authors gratefully acknowledge the support of the Swedish Research Council under Grant No.\ $621$-$2012$-$3863$}
%\bibliographystyle{alpha}
\bibliographystyle{abbrv}
\bibliography{references2020,TotalVR2012}

% ------------------------------------------------------------------------
\end{document}

%% file: S1_Introduction.tex
\section{Introduction}

Integro-differential equations with unbounded operator coefficients are used to describe viscoelasticity, heat transfer with finite propagation speed,
systems with thermal memory, and acoustic waves in composite media \cite{MR1553521,MR1036731,MR2861663}.

In this paper, we consider spectral properties and spectral enclosures of the second-order Gurtin-Pipkin type of integro-differential equations
\begin{equation}\label{wave}
u_{tt}(t)+T_a u(t)+\int_{0}^{t}K_t(t-s)T_b u(s)  ds=h,\quad  K(t)=\sum_{j=1}^{N}a_j e^{-b_j t}
\end{equation}
where $T_a$ and $T_b$ are unbounded self-adjoint operators with compact resolvents in a separable Hilbert space $\CH$. Assume that the constants in the exponential kernel $K$ satisfy $a_j>0$ and $0<b_1<b_2<\cdots<b_N$. After applying the Laplace transform $\hat f(\lambda)=\int_0^{\infty}f(t)e^{-\lambda t}\,dt$ to \eqref{wave} with homogeneous initial conditions, we formally obtain the symbol of \eqref{wave}, which for $\lambda\in\domS:=\C\setminus\{-b_1,-b_2,\dots,-b_N\}$ takes the form
\begin{equation}\label{eq:T}
 T(\lambda)=\lambda^2+T_a-T_b \hat K(\lambda),\quad \hat K (\lambda)=\sum_{j=1}^{N} \frac{a_j b_j}{\lambda+b_j}.
\end{equation}
In the spectral analysis of \eqref{eq:T} we need to deal with the following three difficulties: (i) the essential spectrum is non-trivial, (ii) the domain of $T$ depend on $\lambda$, and (iii) the spectrum is not real. Spectral properties of other classes of rational operator functions have been studied in e.g. \cite{MR1354980,MR1853352,MR1912418,Tretter2008,MR3543766,MR3802490}. 

The point spectrum of \eqref{eq:T} and solvability of \eqref{wave} when $T_b$ is a positive multiple of $T_a$ were studied extensively; see \cite{MR2861663,MR3963333} and the references therein. Those papers derive results for general kernels $K$, using methods from analytic function theory. 

An alternative approach is to base the analysis on the \emph{system operator} \cite{MR281400}. Spectral properties of the system operator that corresponds to \eqref{wave} were in \cite{MR3041669} considered in the space $L^2(0,1)$ with $T_a=a\frac{d}{dx}$, $T_b=b\frac{d}{dx}$ for positive constants $a$ and $b$. A closely related one-dimensional wave equation was studied in \cite{MR3109859}. They derived explicit expressions on the essential spectrum, studied the asymptotics of the eigenvalues, and the Riesz basis property. In all these studies is $T_b=bT_a$ for some positive constant $b$, which implies that \eqref{eq:T} only depends on one unbounded operator. However, there are many applications with two unbounded operator coefficients in \eqref{eq:T}, e.g. contact problems where one material is elastic with no dissipation and one material is viscoelastic \cite{MR1886854}. Moreover, interest for functionally graded materials in elastic structures is growing. Those materials are designed by varying the microstructure from one material to another material with a specific gradient. The gradient is then, on a macroscopic scale, modeled by a smooth function. Under some geometric conditions, it has been shown that the solution of the wave equation with a graded Boltzmann damping decays exponentially to zero \cite{MR1848535}. However, there are to my knowledge no spectral analysis results for this class of problems.

In this paper, we derive spectral results for the operator function \eqref{eq:T}. The results for graded materials include an explicit formula for the essential spectrum and spectral enclosures for the eigenvalues. We consider also the case when $T_b=bT_a$ for some positive constant $b$ and apply the results to wave equations in several space variables with Boltzmann damping.

%, where $\domS:=\C\setminus\{-b_1,-b_2,\dots,-b_N\}$
%The derived spectral properties and spectral enclosures of the operator function $T$ are of particularly interest in the case when $T_b$ is a differential operator with non-constant coefficients. Then, we apply the new spectral results to a wave equation in several space variables with a, in space, graded Boltzmann damping. 

Throughout this article $\mathcal{L}(\CH_1,\CH_2)$ and $\mathcal{B}(\CH_1,\CH_2)$ denote the collection of linear and bounded operators between the Hilbert spaces $\CH_1$ and $\CH_2$, respectively. For convenience we use the notation $\mathcal{L}(\CH_1):=\mathcal{L}(\CH_1,\CH_1)$. Assume that the operator function $T:\CD\rightarrow\mathcal{L}(\CH)$ with the (in general) $\lambda$-dependent domain $\dom T$ is closed. Then, the spectrum of an operator function $T$ is defined as
\[
\sigma (T)=\{\lambda\in\domS \,:\, 0\in\sigma (T(\lambda)) \}.
\] 
The essential spectrum $\sigma_{\text{ess}} (A)$ of a closed operator $A$ is defined as the set of all $\lambda\in\C$ such that $A-\lambda$ is not a Fredholm operator. Hence, the essential spectrum of the operator function $T$ is the set
\[
	\sigma_{\text{ess}} (T)=\{\lambda\in\domS \,:\, 0\in\sigma_{\text{ess}} (T(\lambda)) \}.
\] 
Let $\CD_T, \CD_{\hat T}\subset \C$ and consider the operator functions
\[
T:\CD_T\rightarrow \mathcal{L}(\CH_1,\CH_2),
\quad  \hat T:\CD_{\hat T}\rightarrow \mathcal{L}(\hat \CH_1,\hat \CH_2),
\]
with domains $\dom T$ and $\dom \hat T$, respectively. Then $T$ and $\hat T$ are called \emph{equivalent} on $\CD\subset \CD_T\cap \CD_{\hat T}$ \cite{MR0482317,ET2} if there exists invertible operator functions 
\[
E:\CD\rightarrow \mathcal{B}(\hat \CH_2,\CH_2),
\quad  F:\CD\rightarrow \mathcal{B}(\CH_1,\hat \CH_1)
\]
such that
\[
T(\lambda)=E(\lambda)\hat T(\lambda) F(\lambda),\quad \dom T(\lambda)=F(\lambda)^{-1}\dom \hat T(\lambda).
\]
Furthermore,  $T$ and $\hat T$ are called \emph{equivalent after operator extension} \cite{ET2} if
\[
\begin{bmatrix}
T(\lambda) & \\
0 & W_T(\lambda)
\end{bmatrix}
=
\hat E(\lambda)
\begin{bmatrix}
\hat T(\lambda) & \\
0 & W_{\hat{T}}(\lambda)
\end{bmatrix}
\hat F(\lambda)                  
\]
for some bounded operator functions $\hat E$ and $\hat F$ with bounded inverse. Equivalent operators share many properties. In particular, $T(\lambda)$ is closed (closable) if and only if $\hat T(\lambda)$ is closed (closable). Let $\overline T$ and $\overline{\hat{T}}$ denote the closure of two operator functions that are equivalent on $\CD$. Then, if we restrict the functions to $\CD$, the spectrum is preserved:
\[
\sigma_p(\overline T)=\sigma_p(\overline{\hat{T}}),\quad  \sigma_c(\overline T)=\sigma_c(\overline{\hat{T}}),\quad \sigma_r(\overline T)=\sigma_r(\overline{\hat{T}}),
\]
where $\sigma_p$, $\sigma_c$, and $\sigma_r$ denote the point, continuous, and residual spectrum, respectively.

Assume that the operator function 
$\tilde T:\domS\rightarrow\mathcal{L}(\CH\oplus \hat \CH)$ can be represented as
\[
\tilde T(\lambda)=\begin{bmatrix}
A(\lambda) & B(\lambda)\\
C(\lambda) & D(\lambda)
\end{bmatrix}.
\]
The domain of $\tilde T$ is, unless otherwise stated, the \emph{natural domain} of the operator function, $\dom \tilde T(\lambda)=\dom A(\lambda)\cap\dom C(\lambda)\oplus\dom B(\lambda)\cap\dom D(\lambda)$.

%From  \cite{MR3884266}.
%is known to generate a (linear) contraction
%semigroup S(t) = e tA acting on a suitable Hilbert space accounting for the presence of the memory. In the last thirty years, the
%stability properties of S(t) have been widely investigated, and are nowadays well-understood. Depending on the memory kernel
%μ, the semigroup S(t) can be uniformly stable, stable but not uniformly stable, or even not stable (see e.g. [4,9,11,13–15]).
%Nevertheless, despite their intimate connections with the asymptotic analysis of S(t), results dealing with the structure of the
%spectrum of the infinitesimal generator A are rather poorly covered. To the best of our knowledge, the only achievement in this
%direction has been obtained by K. Liu and Z. Liu in [13, Lemma 2.2].

%% file: S2_spectrum.tex
\section{The spectrum and numerical range of $T$}

In this section, we prove theorems that in the application part of the paper are used to study wave equations with graded material properties and Boltzmann damping. For that reason, we introduce in the theorem below a continuous function $b$ on a bounded and sufficiently smooth domain $\Omega\subset\R^n$. This function corresponds in the application to a graded material. Moreover, $1-b(x)\hat K(0)>0$ for any $x\in\overline{\Omega}$ in the studied physical problem and we will therefore restrict the analysis to functions that satisfy the inequality.

The operator function $T:\CD\rightarrow\mathcal{L}(\CH)$ in \eqref{eq:T} is always defined on the maximal domain $\domS:=\C\setminus\{-b_1,-b_2,\dots,-b_N\}$. 
\begin{thm}\label{th:essspec}
	Assume that $L:\CH\rightarrow \CH$ is a $T_a$-compact operator, where $T_a$ is a positive self-adjoint operator and $T_a^{-1}$ is compact. Furthermore, \eqref{eq:T} is a closed operator function in $\CH$, where the self-adjoint operator $T_b$ can be written in the form $T_b=bT_a+L$ for some $b\in C(\bar\Omega)$. Assume $0\leq b_{\text{min}}\leq b(x) \leq b_{\text{max}}$, for $x\in\overline{\Omega}$ and $1>b_{\text{max}}\sum_{j=1}^N a_j$ and let  $f(\lambda,x):=1-b(x)\hat K(\lambda)$. Then
	\begin{description}
		\item[(i)] $\sigma_{\text{ess}}(T)\subset (-b_N,0)$ and
		\[
			\sigma_{\text{ess}}(T)=\overline{\{\lambda\in\C\,:\, f(\lambda,x)=0\,\, \text{for some}\,\, x\in\bar\Omega\}},
		\]
		\item[(ii)] $\sigma_r (T)=\emptyset$.
	\end{description}
\end{thm}
\begin{proof}
	(i) Define the operator function $\hat T:\CD\rightarrow\mathcal{B}(\CH)$ as
	\[
	\hat T(\lambda):=T^{-1/2}_aT(\lambda)T^{-1/2}_a=\lambda^2T^{-1}_a+f-\hat K(\lambda)T^{-1/2}_aLT^{-1/2}_a.
	\]
	From the assumption that $L$ is $T_a$-compact it follows that $T^{-1/2}_aLT^{-1/2}_a$ is compact. Hence
	\[
	\sigma_{\text{ess}}(T)=\sigma_{\text{ess}}(\hat T),
	\]
	where $\hat T(\lambda)$ is a Fredholm operator if $f(\lambda,x)$ is non-zero for $x\in\overline{\Omega}$.  Note that $f(\lambda,x_0)\prod_{j=1}^N (\lambda+b_j)$ for a fixed $x_0\in\overline{\Omega}$ is an $N$-th degree polynomial in $\lambda$. Hence, $f(\lambda,x_0)=0$ has at most $N$ solutions. By assumption is $f(0,x_0)>0$, $f(\lambda,x_0)\rightarrow+\infty$, $\lambda\rightarrow -b_j^-$, and $f(\lambda,x_0)\rightarrow-\infty$, $\lambda\rightarrow -b_j^+$. Hence, $f(\lambda,x_0)=0$ has $N$ real solutions.
	Take positive constants $\epsilon$, $\delta$ and $x_0\in\bar\Omega$ such that $f(\lambda,x_0)=0$ and $|f(\lambda,x_0)|<\epsilon$ for all $x\in\Omega_\delta$ with $\Omega_\delta:=\{x\in\overline\Omega\,:\, |x-x_0|<\delta\}$. Let $u_\delta$ denote a smooth function with support in $\Omega_\delta$ and $\|u_\delta\|=1$. Take a weakly convergent sub-sequence  $\{u_{\delta'}\}$ of the bounded sequence  $\{u_{\delta}\}$. Then
	\[
	\|f(\lambda,\cdot)u_{\delta'}\|\leq \epsilon.
	\]
 Hence, $\lambda\in	\sigma_{\text{ess}}(T)\subset (-b_N,0)$. 
 
 (ii) Take $\lambda\in\sigma_r(T)$ and set  $V:=\overline{\Ran T(\lambda)}$. Since, $T$ is self-adjoint it follows from the projection theorem that
 \[
 \{0\}\neq V^{\perp}=\text{Ker}\, T(\lambda)^*=\text{Ker}\, T(\bar\lambda).
 \]
 Hence, $\bar\lambda$ is an eigenvalue of $T$, but (i) implies that $\lambda\in\sigma_r(T)$ is real. The claim follows since $\lambda\in \sigma_r(T)$ is not an eigenvalue.
\end{proof}	

The concept of Jordan chains can be formulated for general bounded analytic operator functions \cite{MR971506}. We study an unbounded operator function but the spectral problem can also as in Theorem \ref{th:essspec} be formulated in terms of a bounded operator function. Let $\hat T(\lambda_0)u_0=0$. The vectors $u_1, u_2,\dots u_{m-1}$ are associated with the eigenvector $u_0$ if
\[
\sum_{k=0}^{j}\frac{1}{k!}\hat T^{(k)}(\lambda_0)u_{j-k}=0,\quad j=1,2,\dots,m-1.
\]
The sequence $\{u_j\}_{j=0}^{m-1}$ is called a Jordan chain of length $m$ and the maximal length of a chain of an eigenvector and associated vectors is called the multiplicity of the eigenvector \cite{MR971506}.
\begin{lem}
	Assume that $T_b=bT_a$ and $\lambda_0\in\sigma_p (T)\cap\R$, where
	\begin{equation}\label{eq:Jordan}
	\frac{2}{\lambda_0}[(bu_0,u_0)\hat K(\lambda_0)-1]-(bu_0,u_0)\hat K'(\lambda_0)\neq 0.	
	\end{equation}
	Then, the length of the corresponding Jordan chain is one.
\end{lem}
\begin{proof}
	Assume that $(\hat T(\lambda_0) u_0,u_0)=0$ where $\|u_0\|=1$. Then  $(\hat T'(\lambda_0)u_0,u_0)\neq 0$ if and only if \eqref{eq:Jordan} holds and the claim follows from \cite[Lemma 30.13]{MR971506}. The argument in \cite{MR971506} is as follows. Assume that $u_1$ is an associated vector to $u_0$ and let $(\cdot,\cdot)$ denote the inner product in $\CH$. Then $\hat T(\lambda_0)u_1=-\hat T'(\lambda_0)u_0$ and
	\[
	(\hat T'(\lambda_0)u_0,u_0)=-(\hat T(\lambda_0)u_1,u_0)=-(u_1,\hat T(\lambda_0)u_0)=0,
	\] 
	where we used that $\hat T(\lambda_0)^*=\hat T(\lambda_0)$ for real $\lambda_0$.
\end{proof}
Note that $(bu_0,u_0)=b$ when $b$ is constant, which simplifies the condition \eqref{eq:Jordan}.

A frequently studied case is when $T_b=bT_a$ for some positive constant $b$. The operator function can then be written in the form
\[
	T(\lambda)=\lambda^2+[1-b\hat K(\lambda)]T_a.
\]
Let $\dom T(\lambda)=\CH$ if $1-b\hat K(\lambda)=0$ and $\dom T(\lambda)=\dom T_a$ otherwise. Assume $1-b\hat K(\lambda)\neq 0$ and set
\begin{equation}\label{eq:G}
	G(\lambda):=-\frac{\lambda^2}{1-b\hat K(\lambda)}.
\end{equation}
Then $\lambda\in\sigma (T)$ if and only if $G(\lambda)\in\sigma (T_a)$ for some $\lambda\in\domS$.
\subsection{The enclosure of the numerical range}

The numerical range of  $T:\CD\rightarrow\mathcal{L}(\CH)$ is the set 
\[
W(T)=\{\lambda\in\mathcal{D} :\,\exists u\in\dom T\setminus \{0\}, \|u\|=1, \text{so that}\, (T(\lambda) u,u)=0\},
\]
where $(\cdot,\cdot)$ and $\|\cdot\|$ denote the inner product and norm in $\CH$, respectively. We assume that $T$ satisfies the assumptions in Theorem \ref{th:essspec} with $T_b=bT_a$. Hence $0\not\in\overline {W(\hat T(0))}$, which implies $\sigma (T)\subset \overline {W(T)}$ \cite[Theorem 26.6]{MR971506}.

Properties of the numerical range, including the number of components, are of fundamental importance, but it is difficult to study this set directly. Therefore, we considered in \cite{ET1} a systematic approach to derive a computable enclosure of $W$ for a particular class of rational operator functions. In this paper, we use a similar idea to derive an enclosure of $\overline{W(T)}$, but we incorporate the special structure with the two (in general) unbounded operator coefficients $T_a$ and $T_b$.
%Define for $\lambda\in\domS$ the form $\sqt (\lambda): L^2(\Omega)\times L^2(\Omega)\rightarrow \C$ by
%\[
%\sqt[u,v](\lambda)=\lambda^2(u,v)+a\sqt_1[u,v]-\sqt_b[u,v]\hat{K}(\lambda),\quad u,v\in\dom \sqt (\lambda).
%\]
Take $u\in\dom T$, with $\|u\|=1$ and set
\[
\alpha_u=(T_a u,u),\quad \beta_u=(T_b u,u),
\]
where $b_{\text{min}}\alpha_u\leq \beta_u\leq b_{\text{max}}\alpha_u$. Hence, $\lambda\in W(T)$ if it exists a normalized vector $u\in\dom T\setminus \{0\}$ such that
\[
 \lambda^2+\alpha_u-\beta_u \hat{K}(\lambda)=0.
\]
The enclosure $W_{\alpha,\beta}(T)$ of $\overline{W(T)}$ is then defined as
\[
W_{\alpha,\beta}(T):=\overline{\{\lambda\in\domS\, :\, \lambda^2+\alpha-\beta\hat{K}(\lambda)=0, \alpha\in W(T_a),\, \beta\in[b_{\text{min}}\alpha,b_{\text{max}}\alpha]\}}.
\]
Note that $\alpha$ only depend on $W(T_a)$, which in many cases is easy to find.
%Then $\overline{W(T)}\subset W_{\alpha,\beta}(T)$. %Hence, $\lambda\in W_{\alpha,\beta}(T)$ if and only if
%\[
%g(\lambda):=-\frac{\lambda^2}{1-\hat b\sum_{j=1}^N\frac{a_j b_j}{\lambda+b_j}}\in W(T_a)
%\]
%for some $\hat b\in [b_{\text{min}},b_{\text{max}}]$.
%Note that $g(\lambda)\rightarrow -\infty$, $\lambda\rightarrow\pm\infty$, which implies that $W_{\alpha,\beta}(T)\cap\R$ is contained in a bounded interval for any self-adjoint operator $T_a$ that is bounded from below.
\begin{lem}\label{lem:ess}
	Assume that $T$ satisfies the assumptions in Theorem \ref{th:essspec} with $T_b=bT_a$ and set
	\[
	c_0=\min_{\lambda\in\R}\left \{g(\lambda)\geq W_{\text{min}}(T_a)\, \text{for some}\,\, \hat b\right\},\, c_1=\max_{\lambda\in\R}\left \{g(\lambda)\geq W_{\text{min}}(T_a)\, \text{for some}\, \hat b\right\},
	\]
	where
	\begin{equation}\label{eq:g}
		g(\lambda):=-\frac{\lambda^2}{1-\hat b\sum_{j=1}^N\frac{a_j b_j}{\lambda+b_j}},\quad \hat b\in [b_{\text{min}},b_{\text{max}}].
	\end{equation}
	Then
	\[
	\sigma_{\text{ess}} (T)\subset  W_{\alpha,\beta}(T)\cap \R\subset [c_0,c_1]\subset (-b_N,0].
	\]
	Moreover, $W_{\alpha,\beta}(T)$ is symmetric with respect to $\R$.
\end{lem}
\begin{proof}
	We have $\lambda\in W_{\alpha,\beta}(T)$ if and only if $g(\lambda)\in W(T_a)$ for some $\hat b$. Note that $g(\lambda)\rightarrow -\infty$, $\lambda\rightarrow\pm\infty$, which implies that $W_{\alpha,\beta}(T)\cap\R$ is contained in a bounded interval. Set
	\begin{equation}\label{eq:r}
	r(\lambda)=\lambda^2+\alpha-\beta\sum_{j=1}^N\frac{a_jb_j}{\lambda+b_j}.	
	\end{equation}
	From the assumption follows $r(0)>0$ and we have $r'(\lambda)>0$ for $\lambda\geq 0$. Moreover, $r(\lambda)>0$ for $\lambda<-b_N$.
	Define the $(N+2)$-degree polynomial $g_{\alpha,\beta}(\lambda)=r(\lambda)\prod_{j=1}^N (\lambda+b_j)$. The polynomial $g_{\alpha,\beta}$ has real coefficients, which implies the symmetry with respect to $\R$.
\end{proof}
In the following lemma, $\C_{-}$ denote the set of all complex numbers with non-positive real part.
\begin{lem}
	Assume that $T$ satisfies the assumptions in Theorem \ref{th:essspec} with $T_b=bT_a$ and set $\lambda:=x+iy$, $x,y\in\R$. Then $W_{\alpha,\beta}(T)\subset \C_{-}$ and $\lambda=x+iy\in 	W_{\alpha,\beta}(T)\setminus\R$ if and only if
	\begin{equation}\label{eq:reim}
	\begin{cases}
	2x+\beta\sum_{j=1}^N\frac{a_jb_j}{(x+b_j)^2+y^2} & =0,\\
	x^2-y^2+\alpha-\beta\sum_{j=1}^N\frac{a_jb_j(x+b_j)}{(x+b_j)^2+y^2} & =0,	
	\end{cases}
	\end{equation}
	for some $\beta\in[b_{\text{min}}\alpha,b_{\text{max}}\alpha]$, with $ \alpha\in W(T_a)$. 
\end{lem}
\begin{proof}
	The system \eqref{eq:reim} follows directly by taking the real and imaginary parts of \eqref{eq:r} and $W_{\alpha,\beta}(T)\subset \C_{-}$ is a consequence of \eqref{eq:reim}.
\end{proof}

\begin{lem}
	Set $I_\alpha=\{iy\in i\R\,:\, y=\pm\sqrt{\alpha}\,,\ \text{for some}\,\, \alpha\in W(T_a)\}$. Then
	\[
	W_{\alpha,\beta}(T)\cap i\R=
	\begin{cases}
	\emptyset, & b_{\text{min}}>0,\\	
	I_\alpha,  & b_{\text{min}}=0,
	\end{cases}
	\]
\end{lem}
\begin{proof}
	The claims follow directly from  \eqref{eq:reim}.
\end{proof}
Note that we for many configurations do not expect spectrum on $i\R$ even if $b_{\text{min}}=0$. The reason is that it has been shown that the solution of the corresponding time dependent problem decays exponentially to zero under some additional "geometric" conditions on $b$ \cite{MR1848535}.

 In the next subsection we provide more detailed spectral results for the case with one rational term.
%Note, that the $\lambda$-dependent domain implies that function is not a holomorphic family of type (A) or type (B), which has been extensively studied; See Kato p. 375 and p. 395.

\subsection{A one-pole case}

In this subsection, we consider the rational operator function with one pole. Assume that $T$ satisfies the assumptions in Theorem \ref{th:essspec} with $T_b=bT_a$ and define for $\lambda\in\C$ the operator polynomial
\[
P(\lambda):=(\lambda+b_1)T(\lambda)=\lambda^3+b_1\lambda^2+T_a\lambda+b_1 T_a-a_1 b_1T_b.
\]
Set 
\[
	p_{\alpha,\beta}(\lambda)=\lambda^3+b_1\lambda^2+\alpha\lambda+b_1\alpha-a_1 b_1\beta
\]
and define the enclosure $W_{\alpha,\beta}(P)$ of $\overline{W(P)}$ as
\[
W_{\alpha,\beta}(P):=\overline{\{\lambda\in\C\, :\, p_{\alpha,\beta}(\lambda)=0, \alpha\in W(T_a),\, \beta\in[b_{\text{min}}\alpha,b_{\text{max}}\alpha]\}}.
\]
Then, the numerical range $W(P)$ is composed of at most three components and the following lemma provides, for the one-pole case, more explicit constants than Lemma \ref{lem:ess}.
\begin{lem}\label{lemOne:ess} 
	The enclosure $W_{\alpha,\beta}(P)$ has the following properties
		\begin{description}
		\item[(i)] $W_{\alpha,\beta}(P)\cap\R\subset [c_0,c_1]$, where  $[c_0,c_1]\subset [-b_1,0]$ are given by
		\[
		c_0=\min_{\lambda\in\R}\left \{g(\lambda)\geq W_{\text{min}}(T_a)\, \text{for some}\, \hat b\right\},\quad c_1=-b_1+b_{\text{max}}a_1b_1.
		\]
	\item[(ii)] $[-b_1+b_{\text{min}}a_1 b_1,-b_1+b_{\text{max}}a_1 b_1]\subset W_{\alpha,\beta}(P)$ if and only if $W(T_a)$ is unbounded.
		%\item $-b_1 \in W_{\alpha,\beta}(P)$ if and only if $b_{\text{min}}=0$.
		\item[(iii)] Set
		\[
		d_0=-\frac{b_{\text{max}} a_1 b_1}{2},\quad d_1=-\frac{c_0}{2}-\frac{b_1}{2},\quad \hat d=\sqrt{W_{\text{min}}(T_a)-d_0^2-2d_0c_0}.
		\]
		Then the numerical range is contained in the union of
		\begin{align*} 
			S_0	&=  [c_0,c_1], \\ 
			S_+ &=  \{s\in\C\,:\, \text{Re}\, s\in [d_0,d_1],\, \text{Im}\, s\geq \hat d\}, \\
			S_-	&=\{s\in\C\,:\, \text{Re}\, s\in [d_0,d_1],\, \text{Im}\, s\leq -\hat d\}.
		\end{align*}
	
		%\item Assume $\lambda\in  W_{\alpha}(P)\setminus\R$. Then
		%		\[
		%			\text{Re}\,\lambda\in\left [-\frac{b_{\text{max}}a_1 b_1}{2a},0\right ),\quad \text{Im}\,\lambda\in\R\setminus \{0\}.
		%		\]
		%	and in the limit $\alpha\rightarrow\infty$, we have 
		%	\[
		%		W_{\alpha}(P)=[-\frac{b_{\text{max}}a_1 b_1}{2a},-\frac{b_{\text{min}}a_1 b_1}{2a}]+i[C\alpha,\infty), 
		%	\]
		%	for some $C>0$.
		%\item Assume $\lambda\in  W_{\alpha}(P)\setminus\R$, $\alpha>0$, and Im $\lambda\rightarrow 0$. Then
		%\[
		%		\text{Re}\lambda\leq -w_{\text{min}}-\sqrt{\alpha a+w_{\text{min}}^2}, \quad w_{\text{min}}=-\frac{b_{\text{max}}a_1 b_1}{2a}.
		%\]	
		%\item 
		%\[
		%	|\text{Im}\lambda|^2\geq \min_{\alpha}\left \{\frac{b_1(a+ba_1)}{|c|}\alpha-\left (-w_{\text{min}}-\sqrt{\alpha a+w_{\text{min}}^2}\right )^2\right \}
		%\]
	\end{description}
\end{lem}
\begin{proof}
	We prove the first two claims by studying the function $g$ defined in \eqref{eq:g}. 
	\begin{description}
		\item[(i)]  The claim follows from the assumptions $1-b_{\text{max}}a_1>0$, $W_{\text{min}}(T_a)>0$ and the definition of $g$.
		\item[(ii)] The function $b$ is continuous and $g$ is singular when $\lambda=-b_1+\hat ba_1 b_1$ for some $\hat b\in [b_{\text{min}},b_{\text{max}}]$.
		\item[(iii)] Set $p(\lambda)=p_{\alpha,\beta}(\lambda)$ and write $p$ in the form
		\[
		p(\lambda)=\lambda^3+b_1\lambda^2+\alpha \lambda+c,
		\]
		where $c=b_1(1+ \hat ba_1)\alpha$,  $\hat b\in [b_{\text{min}},b_{\text{max}}]$. This polynomial can be written in the form
		\[
		p(\lambda)=(\lambda-(u+iv))(\lambda-(u-iv))(\lambda-w)
		\]
		for some $u,v,w\in\R$. Identification of the coefficients give
		\begin{align*} 
		b_1 &=  -2u-w, \\ 
		\alpha  &=  u^2+v^2+2uw, \\
		\alpha b_1 (1+\hat ba_1)&=-w(u^2+v^2).
		\end{align*}
		%Hence, $\alpha=0$ implies the solutions $u =v =0$ and $w=-b_1$.
		%\item 
		Hence, $\alpha>0$ implies $w<0$ and $u+iv\neq 0$.
		Since $w$ belongs to the bounded interval $[c_0,c_1]$, it follows that $u=-\frac{1}{2}(b_1+w)$ is bounded and 
		\[
		d_0:=-\frac{1}{2}(b_1+c_1)\leq u\leq -\frac{1}{2}(b_1+c_0)=:d_1.
		\]
		Moreover
		\[
		u^2+v^2=R(\alpha)^2,\quad R(\alpha)^2=\frac{b_1(1+\hat ba_1)}{|w|}\alpha,
		\]
		where $ R(\alpha)^2\rightarrow\infty$, $\alpha\rightarrow\infty$. Hence $|v|\rightarrow\infty$ when $\alpha\rightarrow\infty$.
		We have $v^2=\alpha-u^2-2uw\geq W_{\text{min}}(T_a)-d_0^2-2d_0w_0$.
	\end{description}
\end{proof}	

%% file: S3_equivalence.tex
\section{On equivalence and linearization of $T$}

In this section, we consider the relation between the operator function $T$ in \eqref{eq:T} and the linearization of $T$ in a general setting. Let $\CH$ be a Hilbert space and assume that $T_a$ with domain $\dom T_a$ is self-adjoint in $\CH$ and bounded by $\eta$ from below. Furthermore, $T_b$ with domain $\dom T_b$ is assumed to be non-negative and self-adjoint in $\CH$.

Set $\hat \CH=\hat \CH_1\oplus\hat\CH_1\oplus\cdots\oplus\hat\CH_N$ and let $B=[B_1\,B_2\,\dots\,B_N]: \hat \CH\rightarrow \CH$ denote a densely defined operator. Then $T:\domS\rightarrow\mathcal{L}(\CH)$ can be written in the form
\begin{equation}\label{eq:Schur}
T(\lambda)=T_a+\lambda^2-B(D+\lambda)^{-1}B^*,
\end{equation}
where $B_j=\sqrt{a_jb_j}T_b^{1/2}$ and $D=\text{diag}\,(b_1,b_2,\dots,b_N)$. Let $\tilde\CH=\CH\oplus \hat \CH$ and define the operator function $P:\C\rightarrow\mathcal{L}(\tilde\CH)$ as
\begin{equation}\label{eq:P}
P(\lambda)=
\begin{bmatrix}
T_a+\lambda^2 & B\\
B^* & D+\lambda
\end{bmatrix}
=\begin{bmatrix}
T_a+\lambda^2 		& B_1			& B_2			& \cdots		& B_N\\
B_1^* 				& b_1+\lambda	& 0  			& \cdots		& 0  \\
B_2^*				& 0				& b_2+\lambda	& \cdots		& 0   \\
\vdots				& \vdots		& \vdots		& \ddots		& \vdots\\
B_N^*				& 0				& 0				& \cdots		& b_N+\lambda
\end{bmatrix}
\end{equation}
with the natural domain. This operator function is formally related to $T$ by the relation
\[
\begin{bmatrix}
T(\lambda) & \\
0 & D+\lambda
\end{bmatrix}
=
\begin{bmatrix}
I	& 	-B(D+\lambda)^{-1}\\
0 & I
\end{bmatrix}
\begin{bmatrix}
T_a+\lambda^2 & B\\
B^* & D+\lambda
\end{bmatrix}
\begin{bmatrix}
I	& 	0\\
-(D+\lambda)^{-1}B^* &	I 
\end{bmatrix}
\]
However, the relation is not an equivalence when $B$ is unbounded since the outer operators in the product are not bounded.

Define the operator functions
\[
P_W(\lambda)=
\begin{bmatrix}
T_a+\lambda^2 	& B			& 0\\
B^*		 		& D+\lambda	&  0\\
0		 		& 0			& -\lambda
\end{bmatrix},\quad
\tilde P_W(\lambda)=
\begin{bmatrix}
T_a+\lambda^2 	& 0			& B\\
0		 		& -\lambda	&  0\\
B^*		 		& 0			& D+\lambda
\end{bmatrix}
\]
Those functions are equivalent on $\C$ since
\[
\tilde P_W(\lambda)=E P_W(\lambda) E,\quad E=\begin{bmatrix}
I 	& 0	& 0\\
0	& 0	& I\\
0	& I	& 0
\end{bmatrix}.
\]
Moreover, by applying the theory developed in \cite{ET2}, we obtain the equivalence  $	\tilde P_W(\lambda)=\tilde E(\lambda) \tilde T(\lambda) \tilde F(\lambda)$, where
\begin{equation}\label{eq:tildeT}
\tilde T(\lambda)=
\begin{bmatrix}
-\lambda & -T_a		& -B\\
I		 & -\lambda	&  0\\
0		 & B^*		& D+\lambda
\end{bmatrix},
\end{equation}
and
\[
\tilde E(\lambda)=\begin{bmatrix}
-I 	& -\lambda	& 0\\
0	& -\lambda	& 0\\
0	& 0	& I
\end{bmatrix},\quad \tilde F(\lambda)=\begin{bmatrix}
\lambda 	& I	& 0\\
I	& 0	& 0\\
0	& 0	& I
\end{bmatrix}.
\]
The function $P_W$ can be written in the form
\[
P_W(\lambda)=
\begin{bmatrix}
P(\lambda) & 0\\
0	& W(\lambda)
\end{bmatrix},\quad 	W(\lambda)=-\lambda
\]
and we have therefore shown that $P$ after operator function extension with $W$ is equivalent to $\tilde T(\lambda)$. Hence, under the assumption that the operator functions are closable we have the following properties
\[
\sigma_p(\overline{\tilde{T}})\setminus\{0\}=\sigma_p(\overline P)\setminus\{0\},\quad  \sigma_c(\overline{\tilde{T}})\setminus\{0\}=\sigma_c(\overline P)\setminus\{0\},\quad \sigma_r(\overline{\tilde{T}})\setminus\{0\}=\sigma_r(\overline P)\setminus\{0\}.
\]
\begin{thm}
	Assume that $T_a$ with domain $\dom T_a$ is self-adjoint in $\CH$ and bounded by $\eta$ from below. Let $D$ with domain $\dom D$ denote a self-adjoint operator in $\hat\CH$ and assume that $B:\hat\CH\rightarrow\CH$ is a densely defined closed operator with $\dom B\subset \dom D$, where $\dom B$ is a core of $D$. Moreover, we assume that $\dom |T_a|^{1/2}\subset\dom B^*$.
	
	Then, the operator functions $P:\C\rightarrow\mathcal{L}(\tilde\CH)$ and $\tilde T:\C\rightarrow\mathcal{L}(\CH\oplus\tilde\CH)$ are closable.
\end{thm}
\begin{proof}
	Write $P$ in the form
	\begin{equation}
	P(\lambda)=\mathcal{A}+\mathcal{B}(\lambda)
	\end{equation}
	where
	\begin{equation}
	\mathcal{A}=
	\begin{bmatrix}
	T_a & B\\
	B^* & D
	\end{bmatrix},\quad
	\mathcal{B}(\lambda)=
	\begin{bmatrix}
	\lambda^2 & 0 \\
	0 & \lambda
	\end{bmatrix}.
	\end{equation}
	Hence, $P(\lambda)$ is a bounded perturbation of a symmetric and upper-dominant block operator matrix, which implies that $P$ is closable \cite{MR1380703}. Let $\overline{\mathcal{A}}$ denote the closure of $\mathcal{A}$. Then, the self-adjoint operator $\overline{\mathcal{A}}$ is given by
	\[
	\overline{\mathcal{A}}
	\begin{bmatrix}
	u \\
	v
	\end{bmatrix}=
	\begin{bmatrix}
	T_a(u+\overline{(T_a-\eta)^{-1}B}v)-\eta\overline{(T_a-\eta)^{-1}B}v \\
	B^*u+Dv
	\end{bmatrix}
	\]
	with
	\[
	\dom \overline{\mathcal{A}}= \{
	\begin{bmatrix}
	u \\
	v
	\end{bmatrix}\,:\, v\in\dom D,\, u+\overline{(T_a-\eta)^{-1}B}v\in \dom T_a\}.
	\]
	Hence, $\overline{P}(\lambda)=\overline{\mathcal{A}}+\mathcal{B}(\lambda)$ with domain $\dom \overline{\mathcal{A}}$ is closed. Since $P$ is closable it follows from the equivalence discussed above that $\tilde T$ is closable.
\end{proof}
Assume that $T_b$ is bounded by zero from below, $T_a>\mu>0$, and $T_b$ is $T_a$-bounded. Then, we can as an alternative to \eqref{eq:Schur} consider the bounded operator function
\begin{equation}\label{eq:Schur2}
\hat T(\lambda)=I+\lambda^2T_a^{-1}-\hat B(D+\lambda)^{-1}\hat B^*,
\end{equation}
where $\hat B=[\hat B_1\,\hat B_2\,\dots\,\hat B_N]$, $\hat B_j=\sqrt{a_jb_j}(T_a^{-1}T_b)^{1/2}$ and $D=\text{diag}\,(b_1,b_2,\dots,b_N)$. The rational operator function $\hat T$ is then, after extension, equivalent to a block operator matrix on the form \eqref{eq:P} with bounded entries.

\subsection{One pole}

Let $\tilde\CH=\CH\oplus \hat\CH_1$ and define the closed operator $B: \hat\CH_1\rightarrow \CH$ such that $BB^*v_1=T_b v_1$. Assume that the block operator function $P:\C\rightarrow \mathcal{L}(\tilde\CH)$,
\begin{equation}\label{eq:POne}
P(\lambda)=
\begin{bmatrix}
T_a+\lambda^2 & B\\
B^* & b_1+\lambda
\end{bmatrix}
\end{equation}
with $\dom P(\lambda) =\dom T_a \oplus \dom B$ for $\lambda\in\C$ is closed.
\begin{prop} 
	Assume $B$ is injective. Then $-b_1\not\in\sigma_p(P)$.
\end{prop}
\begin{proof}
	Assume  $-b_1\in\sigma_p(P)$. Then it exists a non-zero $v=[v_1,v_2]^T\in \tilde\CH$ such that
	\begin{equation}\label{eq:b1eq1}
	(T_a+b_1^2)v_1+Bv_2=0,
	\end{equation}
	\begin{equation}\label{eq:b1eq2}
	B^*v_1=0.
	\end{equation}
	Take $v_1=0$. Then $B v_2=0$, which since $B$ is injective implies $v_2=0$. Assume $v_1\neq 0$. Taking the inner product with $v_1$ in \eqref{eq:b1eq1} we find that
	\[
	0=( (T_a+b_1^2)v_1,v_1)+(v_1,B^* v_1)
	\]
	but \eqref{eq:b1eq2} implies that $B^* v_1=0$ and $T_a+b_1^2\geq b_1^2$.
\end{proof}
\begin{prop}\label{prop:P} 
	Let $P$ denote the operator function \eqref{eq:POne}. Then
	\begin{description}
		\item[(i)] If $\lambda\in\sigma_p(P)\setminus\{-b_1\}$ with eigenvector $v=[v_1\, v_2]^T$. Then $v_1$ is an eigenvector of $T$ at $\lambda$.
		\item[(ii)] If $\lambda\in\sigma_p(T)$ with eigenvector $v_1$. Then
		\[
			v=
			\begin{bmatrix}
				v_1\\
				-(b_1+\lambda)^{-1}B^*v_1
			\end{bmatrix}
		\]
		is an eigenvector of $P$ at $\lambda$.
	\end{description}
\end{prop}
\begin{proof}
	(i)  Assume $\lambda\in\sigma_p(P)\setminus\{-b_1\}$ with eigenvector $v=[v_1\, v_2]^T$. Then
	\begin{equation}\label{eq:nb1eq1}
	(T_a+\lambda^2)v_1+Bv_2=0,
	\end{equation}
	\begin{equation}\label{eq:nb1eq2}
	B^*v_1+(b_1+\lambda)v_2=0.
	\end{equation}
	Hence \eqref{eq:nb1eq1} with $v_2=-(b_1+\lambda)^{-1}B^* v_1$ show that $T(\lambda)v_1=0$.\\\hspace{5mm}
	(ii) The claim follows immediately from the definition of $P(\lambda)$. 
\end{proof}
Let $\tilde T:\C\rightarrow\mathcal{L}(\CH\oplus\tilde\CH)$ denote the operator function \eqref{eq:tildeT} in the one-pole case. Hence, $\tilde\CH=\CH\oplus \hat\CH_1$ and
\begin{equation}\label{eq:OnetildeT}
\tilde T(\lambda)=
\begin{bmatrix}
-\lambda & -T_a		& -B\\
I		 & -\lambda	&  0\\
0		 & B^*		& b_1+\lambda
\end{bmatrix},
\end{equation}
with $\dom \tilde T(\lambda) =\CH\oplus\dom T_a \oplus \dom B$ for $\lambda\in\C$.
\begin{prop}\label{prop:tildeT}
	Let $P$ and $\tilde T$ denote the operator functions in \eqref{eq:POne} and \eqref{eq:OnetildeT}. Then
	\begin{description}
		\item[(i)] If $\lambda\in\sigma_p(\tilde T)$ with eigenvector 
		$v=[v_1\, v_2\, v_3]^T$. Then $[v_2\,v_3]^T$ is an eigenvector of $P$ at $\lambda$.
		\item[(ii)] If $\lambda\in\sigma_p(P)$ with eigenvector $[v_2\,v_3]^T$. Then $v=[\lambda v_2\, v_2\, v_3]^T$ is an eigenvector of $\tilde T$ at $\lambda$.
	\end{description}
\end{prop}
\begin{proof}
	The claims follow from the definitions of $P$ and $\tilde T$. 
\end{proof}
We can write the rational operator function \eqref{eq:Schur} in the alternative form
\[
T(\lambda)=T_a+\lambda^2-\hat B(b_1+\lambda)^{-1}\hat C,\quad T_a=\hat B\hat C
\]
and use the approach in \cite{ET2} to associate an operator pencil in $\CH\oplus\tilde\CH$. By switching order of the first two spaces in $\CH\oplus\tilde\CH=\CH\oplus\CH\oplus\hat\CH_1$, we obtain
\[
\tilde T(\lambda)=
\begin{bmatrix}
-\lambda 			& 	I		& 0        \\
-T_a	 	& -\lambda	&  -\hat B \\
\hat C		& 	0		& b_1+\lambda
\end{bmatrix},
\]
with $\dom \tilde T(\lambda) =\dom T_a \oplus \CH\oplus \dom \hat B$ for $\lambda\in\C$. The computation
\begin{equation}\label{system}
\begin{bmatrix}
I 	& 	0	& 0\\
0   &   I	& 0\\
0	& 	0	& -I
\end{bmatrix}
\begin{bmatrix}
-\lambda 			& 	I		& 0        \\
-T_a	 	& -\lambda	&  -\hat B \\
\hat C		& 	0		& b_1+\lambda
\end{bmatrix}
=
\begin{bmatrix}
0 					& 	I	& 0\\
-T_a		 		&   0	&  -\hat B\\
-\hat C	& 	0	& -b_1
\end{bmatrix}-\lambda I_H=:\tilde S-\lambda
\end{equation}
shows that $\tilde T$ is equivalent to $\tilde S-\lambda$, where $\tilde S$ is an abstract version of the system operator considered in \cite{MR3041669}. The generalization to arbitrary rational terms is straightforward. The system operator in \cite{MR3041669} was derived by introducing new variables and we showed 
above how the system operator is related to the rational operator function $T$.

%This relation will be further studied in Section \ref{}.

%% file: S4_Applications.tex
\section{The elastic wave equation with Boltzmann damping}

Viscoelastic materials such as synthetic polymers, wood, tissue, and metals at high temperature are common in nature and engineering. They are used to absorb chock, to isolate from vibrations, and to dampen noise. In this section we study the wave equation in $\Omega\subset\R^n$ with Boltzmann damping \cite{MR1036731}. %Hence, we assume that the stress $\sigma$ and the strain $\epsilon$ is related thorough the Boltzmann integral
%\[
%\sigma(x,t)=\int_0^t\eta(x,t-s)\epsilon_t(x,s)\,ds.
%\]
%This is a constitutive relation that models time delay as well as energy loss. 
The wave equation with Boltzmann damping can for $u=u(t): [0,\infty )\rightarrow L^2(\Omega)$ be written in the form
\begin{equation}\label{wave Bolzmann}
u_{tt}(t)-a\Delta u(t)+a\int_{0}^{t}K_t(t-s)\Div (b\nabla u)(s)  ds=0,
\end{equation}
with $u(x,t)=0$, $x\in\partial\Omega$ and $u(x,0)=u_0(x)$, $u_t(x,0)=u_1(x)$ \cite{MR1036731,MR1886854}.

Let $b$ be a $L^{\infty}$-function on a bounded $\Omega\subset \R^n$ with a Lipschitz continuous boundary $\partial\Omega$ and let $a$ denote a positive constant. The operator $T_b: L^{2}(\Omega)\rightarrow L^{2}(\Omega)$ is taken as the self-adjoint operator whose quadratic form is the closure of the form 
\[
\sqt_b[u,v]= a\int_{\Omega}b \nabla u \cdot \nabla vdx
\] 
with domain $C_0^{\infty}(\Omega)$. Similarly, $T_a$ denotes the self-adjoint realization of $-a\Delta$; see \cite[p. 331]{Kato1980}.

\subsection*{Space constant Boltzmann damping}

Spectral analysis of \eqref{wave Bolzmann} was in \cite{MR3041669} presented for the $\Omega\subset\R$ and $b$ constant. In the following, we consider the generalization to $\Omega\subset\R^n$, $n=1,2,3$ with one rational term and derive enclosures of the numerical range. Set
\[
\mathcal{H}:=L^2(\Omega)\oplus L^2(\Omega)\oplus L^2(\Omega)^n.
\]
The analog system operator \eqref{system} in $\mathcal{H}$ is then
\[
\tilde S=
\begin{bmatrix}
0 					& 	I	& 0\\
-T_a		 		&   0	&  -\hat B\\
-\hat C	& 	0	& -b_1
\end{bmatrix},\quad T_a=-a\Delta,\quad	 \hat B=a\text{div}b, \quad \hat C=-a_1 b_1\nabla,
\]
with $\dom\tilde S =\dom T_a \oplus L^2(\Omega) \oplus \dom \hat B$.
Moreover, in cases where $\sigma(T_a)$ is known explicitly and $b$ is constant, it is possible to numerically compute eigenvalues of $T$ to very high precision. 
Let $\Omega=(0,l_1)\times (0,l_2)\times\dots\times (0,l_n)$. Then 
\begin{equation}\label{eq:specTa}
\sigma(T_a)=\{\lambda\in\R\,:\: \lambda=\pi^2\sum_{j=1}^n\frac{m_j^2}{l_j^2} \text{for some } m_j\in \mathcal{Z}^{+}\},
\end{equation}
and $\lambda\in\sigma(T)$ if and only if $G(\lambda)\in\sigma (T_a)$, where $G$ is defined in \eqref{eq:G}. Hence, we can numerically determine eigenvalues of $T$ by determining the roots of the polynomial obtained after multiplication by $\lambda+b_1$. For the computations we use the Matlab function \emph{roots}, which for small $N$ fast and robust computes the eigenvalues of the companion matrix.

Let $N=1$, $\Omega=(0,1)\times (0,4)$, $a=2$, $b=1$, $a_1=0.9$, $b_1=0.5$. From Theorem \ref{th:essspec} with
\[
f(\lambda)=1-b\frac{a_1b_1}{\lambda+b_1}
\]
follows that $\sigma_{\text{ess}}(T)={c_1}$ and Lemma \ref{lemOne:ess} gives $c_1=-b_1+ba_1b_1$. Note that $T(c_1)=c_1^2$, $\dom T(c_1)=L^2(\Omega)$. Moreover, we obtain
\[
c_0=\min_{\lambda\in\R}\left\{g(\lambda)\geq  \pi^2+\pi^2/16\right\}\approx -0.2759.
\]
The numerically computed spectrum and enclosure of the spectrum are given in Figure \ref{fig:Case1A}, where $W_{\alpha,\beta}(T)=S_0\cup S_+\cup S_{-}$ for 
	\begin{align}\label{eq:result}
S_0	&= [-0.2759,-0.2750],\nonumber \\ 
S_+ &=  \{s\in\C\,:\, \text{Re}\, s\in [-0.11250,-0.11205],\, \text{Im}\, s\geq 4.5715\}, \\
S_-	&=\{s\in\C\,:\, \text{Re}\, s\in [-0.11250,-0.11205],\, \text{Im}\, s\leq -4.5715\}.\nonumber
\end{align}
\begin{figure}
	\includegraphics[width=4cm]{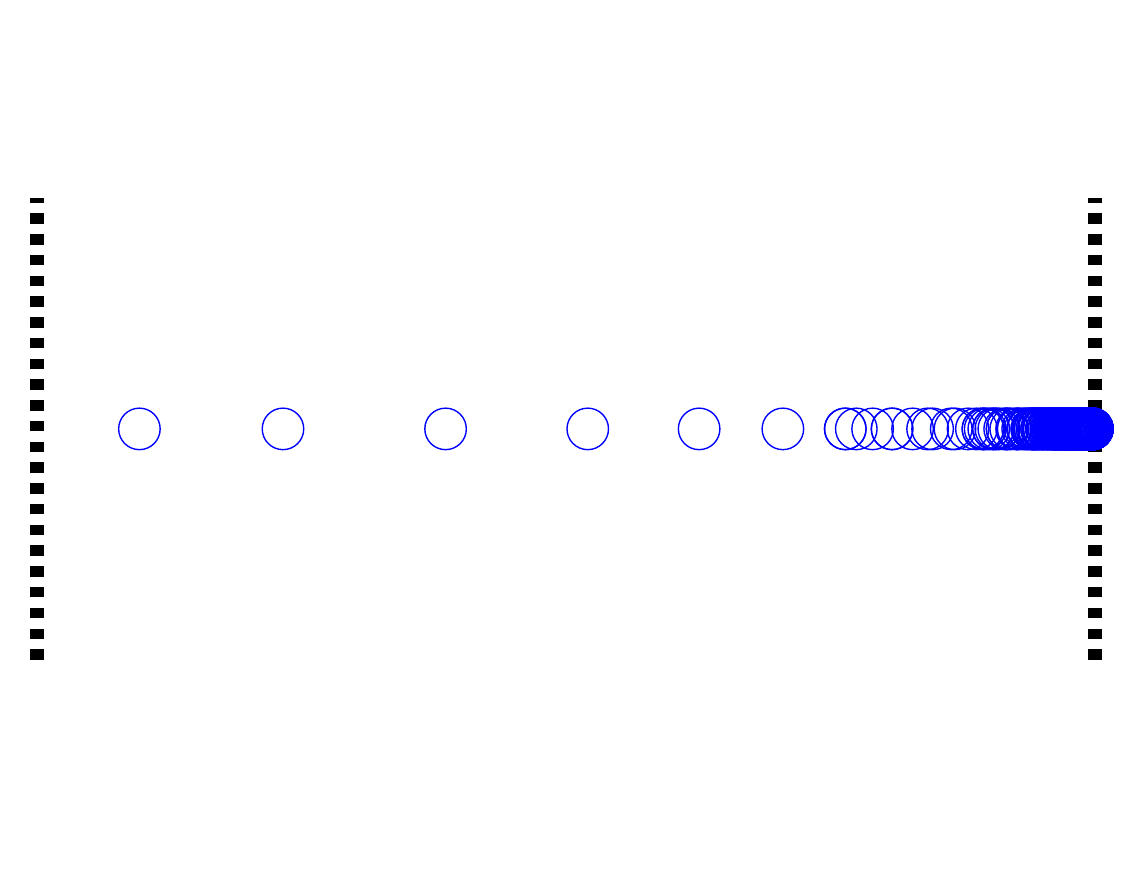}\hspace{2cm}\includegraphics[width=5cm]{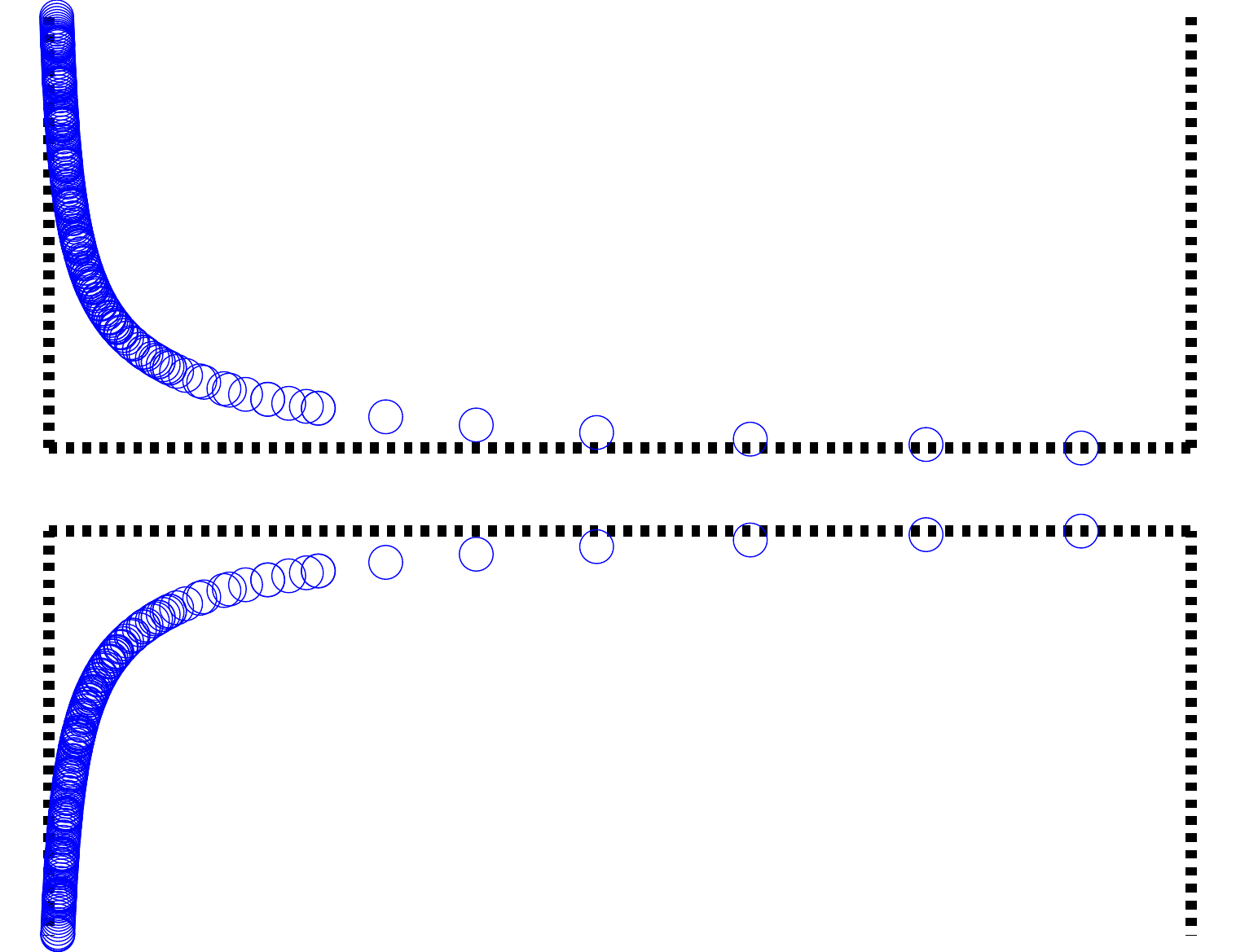}
	\caption{The dotted black lines depicts the boundaries of $S_{0}$, $S_{-}$, and $S_{+}$ in \eqref{eq:result} and the blue circles are eigenvalues with  $|\text{Im}\,\lambda|\leq 50$. Left: Real eigenvalues accumulate at $c_1$. Right: The eigenvalues accumulate at $d_0\pm i\infty$.}
	\label{fig:Case1A}
\end{figure}
Let $d_0$ denote the constant in Lemma \ref{lemOne:ess}. By taking the imaginary part of \eqref{eq:G} with $N=1$, we conclude that a sequence of eigenvalues $\{\lambda_n\}_{n=1}^{\infty}$ with Im $\lambda_n\rightarrow\pm\infty$ only is possible if Re $\lambda_n\rightarrow d_0$ when $n\rightarrow\infty$.
In the case $\Omega=(0,1)$,  it was proved in \cite{MR3041669} that it exists a branch of eigenvalues $\lambda_n$ with $\lambda_n=d_0\pm i\sqrt{a}n\pi+O(n^{-1})$, $n\rightarrow \infty$. Hence, the constant $d_0$ in Lemma \ref{lemOne:ess} is optimal. 

The condition \eqref{eq:Jordan} is satisfied, which implies that eigenvectors corresponding to real eigenvalues have no associated vectors. The eigenvectors of $T$ are related to $P$ (Proposition \ref{prop:P}), $\tilde T$  (Proposition \ref{prop:tildeT}), and $\tilde S$ \eqref{system}.

\subsection*{Graded Boltzmann damping}

Let $T_b=-a\Div b\nabla$, where $b\in C^{\infty}(\overline\Omega)$ is non-negative and $\Omega\subset\R^n$ has smooth boundary. Define
\[
	T_b=bT_a+L,\quad T_a=-a\Delta, \quad L=-a\sum_{j=1}^n\frac{\partial b}{\partial x_j}\frac{\partial}{\partial x_j},
\]	
with $\dom T_a=H^2(\Omega)\cap H_0^1(\Omega)$,  $\dom L=H_0^1(\Omega)$. From Theorem \ref{th:essspec} follows
\[
	\sigma_{\text{ess}}(T)=\overline{\{\lambda\in\C\,:\, f(\lambda,x)=0\,\, \text{for some}\,\, x\in\bar\Omega\}}.
\]
Example. Take, $N=1$, $a=a_1=b_1=1$, $\Omega=(0,1)^2$, and 
\[
b(x_1,x_2)=b_{\text{max}}-2(b_{\text{max}}-b_{\text{min}})[(x_1-\frac{1}{2})^2+(x_2-\frac{1}{2})^2],	
\] 
with $b_{\text{max}}<1$. Then Theorem \ref{th:essspec} implies
\[
\sigma_{\text{ess}}(T)=-1+b(\bar\Omega)=[-1+b_{\text{min}},-1+b_{\text{max}}].
\]
Note that $c_1=-1+b_{\text{max}}$. Take $T_a$ as in \eqref{eq:specTa}, $b_{\text{min}}=0.5$, and $b_{\text{max}}=0.75$. Then
$\sigma_{\text{ess}}(T)=[-0.5,-0.25]$ and
\[
c_0=\min_{\lambda\in\R}\left\{g(\lambda)\geq  2\pi^2\right\}\approx -0.506413. %Reduce[-x^2/(1-0.5/(x+1))>= 2*Pi^2,x]
\]
From Lemma \ref{lemOne:ess} we obtain $W_{\alpha,\beta}(T)=S_0\cup S_+\cup S_{-}$ for 
\begin{align*} 
S_0	&= [-0.506,-0.25], \\ 
S_+ &=  \{s\in\C\,:\, \text{Re}\, s\in [-0.375,-0.246],\,\text{Im}\, s\geq 4.383\}, \\
S_-	&=\{s\in\C\,:\, \text{Re}\, s\in [-0.375,-0.246],\, \text{Im}\, s\leq -4.383\}.
\end{align*}
In the general case with $N$ terms, it follows from Theorem \ref{th:essspec} that the essential spectrum consists of at most $N$ intervals. Take for example $N=2$, $a_2=0.2$, $b_2=1.5$, and $a=a_1=b_1=1$. Then, 
\[
	\sigma_{\text{ess}}(T)=[-1.41,-1.42]\cup [-0.40,-0.067],
\]
when $b_{\text{min}}=0.5$, $b_{\text{max}}=0.75$.

%c0=-0.506413 c1=-0.25 u0=-0.355 u1=-0.2467935 v0=4.3839222